\documentclass[11pt]{amsart}
\usepackage[letterpaper, margin=1in]{geometry}

\usepackage[colorlinks,
            linkcolor=blue,
            anchorcolor=blue,
            citecolor=blue]{hyperref}
            
\usepackage[utf8]{inputenc}
\usepackage{amsmath, amsthm, amssymb, color, mathbbol}
\usepackage[shortlabels]{enumitem}
\usepackage{graphicx}
\usepackage{makecell}
\usepackage[ruled]{algorithm2e}
\usepackage{multicol}
\usepackage[table,xcdraw]{xcolor}
\usepackage{array,amscd}
\usepackage{amsfonts}
\usepackage[T1]{fontenc} 
\usepackage{appendix}
\usepackage[arrow, matrix, curve]{xy}
\usepackage[square,comma,sort,numbers]{natbib}
\usepackage{booktabs}
\usepackage{siunitx}
\usepackage{multirow}
\usepackage[font=small,labelfont=bf]{caption}
\usepackage{subfig}
\usepackage{graphicx, array,amscd}
\usepackage{amsfonts}
\usepackage{bbm}
\usepackage[T1]{fontenc}
\usepackage[square,comma,sort,numbers]{natbib}
\usepackage{booktabs}

\usepackage{blkarray}

\def\ci{\perp\kern-1.3ex\perp}
\def\nci{\not\kern-0.3ex\ci}

\usepackage{tikz}
\usetikzlibrary{snakes}
\usetikzlibrary{shapes}
\usetikzlibrary{arrows}
\usetikzlibrary{positioning}
\usetikzlibrary{calc}
\usetikzlibrary{shapes.geometric}

\setenumerate{leftmargin=*}
\allowdisplaybreaks[3]	

\newtheorem{theorem}{Theorem}
\newtheorem*{theorem*}{Theorem}

\newtheorem{prop}[theorem]{Proposition}

\theoremstyle{definition}
\newtheorem{defn}[theorem]{Definition}
\newtheorem{ex}[theorem]{Example}

\numberwithin{theorem}{section}

\newcommand{\zz}{\mathbb{Z}}

\newcommand{\rr}{\mathbb{R}}
\newcommand{\cc}{\mathbb{C}}

\newcommand{\bfe}{\mathbf{e}}

\newcommand{\bfx}{\mathbf{x}}
\newcommand{\bfy}{\mathbf{y}}

\newcommand{\cm}{\mathcal{M}}

\newcommand{\ii}{\mathcal{I}}

\newcommand{\cp}{\mathfrak{P}}

\newcommand{\ind}{\mbox{$\perp \kern-5.5pt \perp$}}

\newcommand{\lca}{\mathrm{mrca}}
\newcommand{\ones}{\mathbf{1}}
\renewcommand{\top}{\mathrm{top}}
\renewcommand{\int}{\mathrm{Int}}

\usepackage{xcolor}
\definecolor{janepurple}{RGB}{180, 0, 240}
\definecolor{benred}{RGB}{240, 0, 0}

\title[Identifiability of the Rooted Tree Parameters in the CFN-MC Model]{Identifiability of the Rooted Tree Parameter under the Cavender-Farris-Neyman Model with a Molecular Clock}
\author{Jane Ivy Coons and Benjamin Hollering}

\begin{document}
\begin{abstract}
Identifiability of the 
discrete tree parameter is a key property for phylogenetic models since it is necessary for statistically consistent estimation of the tree from sequence data. Algebraic methods have proven to be very effective at showing that tree and network parameters of phylogenetic models are identifiable, especially when the underlying models are group-based. However, since group-based models are time-reversible, only the unrooted tree topology is identifiable and the location of the root is not. In this note we show that the rooted tree parameter of the Cavender-Farris-Neyman Model with a Molecular Clock is generically identifiable by using the invariants of the model which were characterized by Coons and Sullivant. 
\end{abstract}

\maketitle
\section{Introduction and Preliminaries}
The parameters of a statistical model are identifiable if they are uniquely determined by the probability distribution they produce. Identifiability is a key property since it is needed to do statistically consistent inference. In phylogenetics, the identifiability of the tree (or network) parameter is often of particular interest since it is needed to reconstruct evolutionary histories from gene sequence data.  

The identifiability of the tree parameter in many simple phylogenetic models has already been established \cite{chang1996full, steel1998} and many recent results have been obtained for more complicated models including mixtures \cite{allman2010identifiability, hollering2021identifiability, rhodes2012identifiability, long2015identifiability}, networks, \cite{ardiyansyah2021distinguishing, gross2018distinguishing, gross2020distinguishing}, and coalescent models \cite{allman2019species, allman2019nanuq, allman2022identifiability, solis2016inferring,yu2015maximum,zhang2018bayesian}. 

In many of the above works, algebraic methods have been used to show that the tree or network parameters are identifiable. Algebraic techniques are particularly effective for group-based models such as the Cavender-Farris-Neyamn (CFN), Jukes-Cantor, Kimura 2-parameter, and Kimura 3-parameter models. This is because the discrete Fourier transform can be used to simplify the parameterization of group-based models \cite{hendy1996complete, evans1993invariants}.  After applying this linear change of coordinates, many group-based models become \emph{toric varieties} \cite{sturmfels2005toric} and thus the \emph{invariants} of the model become much easier to study and characterize. While the standard group-based models are amenable to algebraic methods, they are also \emph{time-reversible}. This means that the unrooted tree parameter is identifiable, but the location of the root is not identifiable under these models \cite{felsenstein1981evolutionary}. 

In this note we study the identifiability of the rooted tree parameter under the Cavender-Farris-Neyman model with a molecular clock which we call the CFN-MC model after \cite{coons2021toric}. This model is still group-based but the molecular clock condition requires that the time elapsed between any leaf and the root is the same; that is, it restricts the model to ultrametric trees and removes the time-reversibility. The molecular clock is an especially important assumption for inferring phylogenetic histories along shorter time-scales, such as in the cases of bacterial and viral evolution.  The invariants of this model along with the combinatorial structure of the associated polytope were completely characterized in \cite{coons2021toric}. We use these invariants here to obtain the following result, which confirms the biological intuition that the molecular clock allows the location of the root to be identified:

\begin{theorem}
\label{thm:Main}
The rooted tree parameter of the Cavender-Farris-Neyman model with a Molecular Clock is generically identifiable for all trees with the same fixed number of leaves. 
\end{theorem}

It is well known to biologists that the molecular clock assumption allows for the identification of the root \cite{yang2012molecular}. In this note we give a simple algebraic proof of this fact. 

The remainder of this note is organized as follows. In Section \ref{sec:CFNMC} and \ref{sec:GenID} we provide some background on the CFN-MC model and the characterization of the invariants given in \cite{coons2021toric}. In Section \ref{sec:GenID} we give a brief overview of generic identifiability and algebraic tools which can be used to prove such results. In Section \ref{sec:RootedTreeID}, we use the invariants to prove our main result. 

\subsection{The Cavender-Farris-Neyman model with a Molecular Clock}
\label{sec:CFNMC}

In this section, we review the Cavender-Farris-Neyman model with a molecular clock (CFN-MC model). This model is described in detail in Section 3 of \cite{coons2021toric}, and we recount only the aspects that are relevant to this note.
Let $T$ be an $n$-leaf rooted, binary tree with labeled edge lengths. For each edge $e$ in $T$, let $t_e > 0$ denote its length. Let $a(e)$ and $d(e)$ be the two vertices of $e$ so that $d(e)$ is a descendent of $a(e)$; that is, $a(e)$ is the vertex of $e$ that is closest to the root.

The Cavender-Farris-Neyman model (or CFN model) is a  point substitution model that describes mutations from purines (adenine and guanine) to pyrimidines (thymine and cytosine) and vice versa at a single aligned site in the genomes of several taxa. We bijectively identify the set of two states with $\zz_2$. The CFN model arises as a two-state continuous time Markov process along the fixed tree $T$.
The rate matrix of this Markov process has the form
\[
Q = \begin{bmatrix}
- \alpha & \alpha \\
\alpha & - \alpha \\
\end{bmatrix}
\]
where $\alpha > 0$ is a parameter that represents the rate of transition from purine to pyrimidine and vice versa. To obtain a transition matrix $M^e$ for an edge $e$ in $T$, we exponentiate the matrix $Qt_e$; that is,
\[
M^e := \exp(Q t_e) = \begin{bmatrix}
    (1+e^{-2\alpha t_e})/2 & (1-e^{-2\alpha t_e})/2 \\
    (1-e^{-2\alpha t_e})/2 & (1+e^{-2\alpha t_e})/2 \\
\end{bmatrix}.
\]
The entries of this matrix are conditional probabilities; in particular, the $i,j$ entry of $M^e$ is the probability the vertex $d(e)$ has state $j$ given that $a(e)$ has state $i$. Let $m^e_0$ denote the diagonal entry of $M^e$ and let $m^e_1$ denote the off-diagonal entry.

Typically the gene sequences of the taxa at the internal nodes of the tree are unknown. So the CFN model is a hidden variable graphical model in which the states at the leaves of the tree are observed and the states at the internal nodes are not. In order to compute the probability of observing a given $0/1$ sequence $\bfx$ at the leaves of the tree, we must marginalize over all possible labelings of the internal nodes of the tree with the elements of $\zz_2$. We assume a uniform distribution of states at the root of $T$. So the probability of observing the states $\bfx \in \zz_2^n$ at the leaves of the tree under the CFN model is
\begin{equation}\label{eq:Probability}
p_T(\bfx) = \frac{1}{2} \sum_{(x_{n+1},\dots, x_{2n-1})\in \zz_2^{n-1}} \prod_{e \in E(T)} M^e(x_{a(e)},x_{d(e)}).
\end{equation}
The \emph{CFN model} is the set of all probability distributions in $\Delta_{2^n-1}$ of this form. Note that $p_T(\bfx)$ is a polynomial in the entries $m^e_0$ and $m^e_1$ of $M^e$. Thus the CFN model is the intersection of an algebraic variety with the probability simplex and we can study it through the lens of algebraic geometry.

The \emph{molecular clock} is the notion in evolutionary biology that the rate of mutation at a fixed site in the genome is roughly the same between different evolutionary lineages over the same period of time. On the mathematical level, the molecular clock condition imposes that all the labeled tree parameters for the model must be \emph{ultrametric}.

\begin{defn}
A tree $T$ with labeled edge lengths $t_e >0$ is \emph{ultrametric} if for any internal node $v$ and any pair of leaves $i$ and $j$ descended from it, the sum of the edge lengths along the path from $v$ to $i$ is equal to that along the path from $v$ to $j$.
\end{defn}

Equivalently, this condition says that the time elapsed along the path from the root to a leaf is the same for all leaves. Ultrametric trees are also often called \emph{equidistant}. 

\begin{defn}
    The \emph{Cavender-Farris-Neyman model with a molecular clock}, or CFN-MC model, is the set of all probability distributions in $\Delta_{2n-1}$ that arise from ultrametric trees in the CFN model.
\end{defn}

The CFN-MC model is a \emph{group-based} with respect to $\zz_2$; that is, we have identified the states of the model with $\zz_2$ in such a way that the entry $M^e(i,j)$ of each transition matrix depends only on the value of $i-j$. We refer the reader to Chapter 15 of \cite{AS18} for an extensive introduction to group-based phylogenetic models. These models have especially useful algebraic properties; in particular, they are toric after a linear change of coordinates known as the \emph{discrete Fourier transform} 
\cite{evans1993invariants, hendy1996complete}. In 
\cite{coons2021toric}, the authors give a parametrization of the toric variety associated to the CFN-MC model on a phylogenetic tree. In order to introduce this parametrization, we require the following definitions.

Let $R = \cc[q_{\bfx} \mid \bfx \in \zz_{2,\mathrm{even}}^n]$ be the polynomial ring in variables $q_{\bfx}$ indexed by length $n$ $0/1$ strings that sum to an even number. Such labelings of the leaves of $T$ are in bijection with systems of disjoint paths that connect the leaves of $T$. In particular, given $\bfx \in \zz_{2,\mathrm{even}}^n$, there is a unique pairing of the leaves $\ell$ of $T$ with $x_\ell = 1$ so that the paths passing between each pair are pairwise disjoint. Let $\cp(\bfx)$ denote this system of disjoint paths. 

Denote by $\int(T)$ the set of $n-1$ internal nodes of $T$. Given a path $P$ between leaves of $T$, we define its \emph{top-most node} to be the node of $P$ that is closest to the root. Finally, for each $\bfx \in \zz_{2,\mathrm{even}}^n$, we define $\top(\bfx)$ to be the set of all top-most nodes of paths in $\cp(\bfx)$.

\begin{ex}
The tree in figure \ref{fig:Tree} satisfies the molecular clock condition if and only if $t_3 = t_4$, $t_5 = t_6$ and $t_1 + t_3 = t_2 + t_5$. Note that these three linear equations imply all others that most hold in order for the tree to be ultrametric. The polynomial ring $R$ is 
\[
\cc[q_{0000}, q_{1100}, q_{1010}, q_{1001}, q_{0110}, q_{0101}, q_{0011}, q_{1111}].
\]
The internal nodes of $T$ are $\int(T) = \{a,b,c\}$, and the top-sets of the even $0/1$ labelings of the leaves are:
\begin{center}
\begin{tabular}{lcl}
    $\top(0000) = \emptyset$ & \qquad \qquad \qquad & $\top(1100) = \{b\}$ \\
    $\top(1010) = \{ a \}$ & \qquad \qquad \qquad & $\top(1001) = \{a \} $ \\
    $\top(0110) = \{ a \}$ & \qquad \qquad \qquad & $\top(0101) = \{a \} $ \\
    $\top(0011) = \{ c \}$ & \qquad \qquad \qquad & $\top(1111) = \{b,c\} $. \\
\end{tabular}
\end{center}

\end{ex}

\begin{figure}
    \centering
    \begin{tikzpicture}
\draw[fill] (0,0) circle [radius = .05];
\node[below] at (0,0) {1};
\draw[fill] (.75,.75) circle [radius = .05];
\draw[fill] (1.5,0) circle [radius = .05];
\node[below] at (1.5,0) {2};
\draw[fill] (2,2) circle [radius = .05];
\draw[fill] (3.25,.75) circle [radius = .05];
\draw[fill] (2.5,0) circle [radius = .05];
\node[below] at (2.5,0) {3};
\draw[fill] (4,0) circle [radius = .05];
\node[below] at (4,0) {4};
\draw (0,0) -- (2,2) -- (4,0);
\draw (1.5,0) -- (.75, .75);
\draw (3.25, .75) -- (2.5, 0);
\node[left] at (1.2,1.2) {$t_1$};
\node[right] at (2.8,1.2) {$t_2$};
\node[left] at (.5,.5) {$t_3$};
\node[right] at (1, .5) {$t_4$};
\node[left] at (3,.5) {$t_5$};
\node[right] at (3.5,.5) {$t_6$};
\node[left] at (2,2) {$a$};
\node[left] at (.75,.75) {$b$};
\node[right] at (3.25, .75) {$c$};
\end{tikzpicture}
    \caption{A rooted binary tree with four leaves and labeled edge lengths.}
    \label{fig:Tree}
\end{figure}
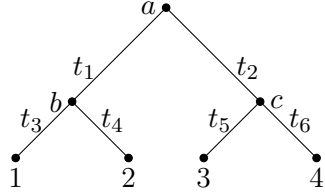

In Section 3 of \cite{coons2021toric}, the authors show that the toric ideal of the CFN-MC model is the kernel of the monomial map,
\begin{align}
    \phi_T : &R  \rightarrow \cc[b_0, b_v \mid v \in \int(T)] \label{eqn:Map}\\
    & q_{\bfx}  \mapsto b_0 \prod_{v \in \top(\bfx)} b_v. \nonumber
\end{align}
We call the kernel of $\phi_T$ the \emph{CFN-MC} ideal and denote it by $I_T$. The polynomials in $I_T$ are the \emph{phylogenetic invariants} of the model.
Note a binomial
\[
\prod_{i=1}^d q_{\bfx_i} - \prod_{i=1}^d q_{\bfy_i}
\]
belongs to $I_T$ if and only if the multiset union, $\Cup_{i=1}^d \top(\bfx_i)$ is equal to $\Cup_{i=1}^d \top(\bfy_i)$. For more information on Markov bases and Gr\"obner bases of toric ideals we refer the reader to \cite{aoki2012markov, sturmfels1996grobner}. 

\begin{ex}
    The toric ideal of phylogenetic invariants $I_T$ for the tree in Figure \ref{fig:Tree} is
    \[
    \langle q_{1010} - q_{1001}, q_{1010} - q_{0110}, q_{1010} - q_{0101}, q_{0000} q_{1111} - q_{1100} q_{0011} \rangle.
    \]
    Note that the linear invariants exactly come from pairs of leaves with the same most recent common ancestor, which we explain further in Proposition \ref{prop:LinearForms}.
\end{ex}



\subsection{Identifiability for Discrete Random Variables}
\label{sec:GenID}
In this section we provide some brief background on generic identifiability of discrete parameters with a particular focus on the identifiability of the tree parameters in phylogenetics. 

Let $\cm_{T}$ denote the CFN-MC model on the rooted tree $T$ after applying the discrete Fourier transform. Let $\mathcal{T}_n$ denote the set of all rooted binary tree topologies on $n$ leaves. Consider the family of models $\{\cm_T\}_{T\in \mathcal{T}_n} = \{\cm_T ~:~ T ~\mathrm{has}~ n ~\mathrm{leaves} \}$ and observe that each model $\cm_T \subseteq \rr^{2^{n-1}}$. In other words, each model lies in the same space. In this setting, the tree parameter $T$ is \emph{globally identifiable} if $\cm_{T^1} \cap \cm_{T^2} = \emptyset$ for every pair of different trees $T^1$ and $T^2$ with $n$ leaves. This will often fail to be true and so the following notion of \emph{generic identifiability} is typically used instead.  
\begin{defn}
\label{def:genID}
Let $\{\cm_T\}_{T \in \mathcal{T}_n}$ be the collection of CFN-MC models for all rooted trees $T$ with $n$ leaves. Then the rooted tree parameter $T$ is \emph{generically identifiable} if for each distinct pair of trees $\{T^1, T^2\}$
\[
\dim(\cm_{T^1} \cap \cm_{T^2}) < \min(\dim(\cm_{T^1}), \dim(\cm_{T^2}))
\]
\end{defn}
This condition is also often called \emph{model distinguishability}. It guarantees that the intersection of $\cm_{T_1} \cap \cm_{T_2}$ is a Lebesgue measure zero subset of each model and thus the probability that a generic data point falls in the intersection is zero. Generic identifiability is also frequently used since algebraic techniques can be applied to study it. The following proposition is a standard algebraic tool for proving generic identifiability results. 

\begin{prop}
\label{prop:idealID}
\cite[Proposition 16.1.12]{AS18}
Let $\cm_1$ and $\cm_2$ be two algebraic models which sit inside the same space. Suppose further that $\ii(\cm_1)$ and $\ii(\cm_2)$ are prime. If there exists polynomials $f_1$ and $f_2$ such that
\[
f_1 \in \mathcal{I}(\cm_1)\setminus \mathcal{I}(\cm_2) ~ \mbox{and}  ~
f_2 \in \mathcal{I}(\cm_2)\setminus \mathcal{I}(\cm_1)
\]
then $\dim(\cm_{1} \cap \cm_{2}) < \min(\dim(\cm_{1}), \dim(\cm_{2}))$.
Furthermore, if $\dim(\cm_1) = \dim(\cm_2)$ and $\ii(\cm_1) \neq \ii(\cm_2)$ then
$\dim(\cm_{1} \cap \cm_{2}) < \min(\dim(\cm_{1}), \dim(\cm_{2}))$. 
\end{prop}
The second statement means that to distinguish two irreducible models of the same dimension, it is enough to find either $f \in \ii(\cm_1) \setminus \ii(\cm_2)$ or $f \in \ii(\cm_2) \setminus \ii(\cm_1)$. We will see in the next section that this is indeed the case for the CFN-MC model.

\section{Identifiability of the Rooted Tree Parameters}
\label{sec:RootedTreeID}
Before we prove the main result, we introduce a few more pieces of notation. 
 For any pair of leaves $i,j \in [n]$, let $\lca(i,j)$ denote the top-most node in the path between them; in other words, this is their most recent common ancestor. Let $\ones_{i,j}$ denote the $0/1$ vector that has $1$ in its $i$th and $j$th positions and $0$ everywhere else. The following two propositions are necessary for our proof of Theorem \ref{thm:Main} and follow directly from the parametrization of $I_T$ given in Equation (\ref{eqn:Map}).

\begin{prop}\label{prop:LinearForms}
    Let $i,j,k,\ell \in [n]$ be (not necessarily distinct) leaves of $T$. The most recent common ancestors $\lca(i,k)$ and $\lca(j,\ell)$ are equal if and only if the linear binomial $q_{\ones_{i,k}} - q_{\ones_{j,\ell}}$ belongs to $I_T$.
\end{prop}

\begin{proof}
    The top-most node in the path from $i$ to $k$ is $\lca(i,k)$ and similarly for the path from $j$ to $\ell$. So the binomial $q_{\ones_{i,k}} - q_{\ones_{j,\ell}}$ lies in the kernel of $\phi_T$ if and only if $\lca(i,k) = \lca(j,\ell)$.
\end{proof}

\begin{prop}\label{prop:Dimension}
For any rooted binary tree $T$ on $n$ leaves, the dimension of $I_T$ is $n-1$.
\end{prop}

\begin{proof}
The singleton containing any internal node $v$ is the top-set of $\ones_{i,j}$ where $i$ and $j$ are any two leaves with $\lca(i,j) = v$. Thus $\bfe_v$ is a column of the matrix $A_T$ that  parametrizes $\cm_T$. So $A_T$ has rank $n-1$, as needed.
\end{proof}

We are now able to prove the main result of this note.

\begin{proof}[Proof of Theorem \ref{thm:Main}]
Let $T^1$ and $T^2$ be distinct rooted, binary phylogenetic trees. For any leaves $a,b \in [n]$, denote by $\lca^i(a,b)$ the most recent common ancestor of $a$ and $b$ in tree $T^i$. For any $\bfx \in \zz^n_{2,\mathrm{even}}$, let $\top^i(\bfx)$ denote its top-set in $T^i$. Since $T^1 \neq T^2$, there exist three leaves $i,j,k \in [n]$ such that the rooted triples $T^1|_{i,j,k}$ and $T^2|_{i,j,k}$ are not equal \cite[Theorem 3]{steel1992}. 

Without loss of generality, suppose that $T^1|_{i,j,k}$ is the rooted triple with cherry ${i,j}$ and that $T^2|_{i,j,k}$ is the rooted triple with cherry ${i,k}$. Then we have $\lca^1(i,k) = \lca^1(j,k)$, but $\lca^2(i,k) \neq \lca^2(j,k)$. Hence, $\top^1(\ones_{i,k}) = \top^1(\ones_{j,k})$, but $\top^2(\ones_{i,k}) \neq \top^2(\ones_{j,k})$. By Proposition \ref{prop:LinearForms}, this implies that the linear binomial $q_{\ones_{i,k}} - q_{\ones_{j,k}}$ belongs to $\ii(\cm_1)$ but not to $\ii(\cm_2)$. By Proposition \ref{prop:Dimension}, $\dim(\ii(\cm_1)) = \dim(\ii(\cm_2))$. Thus by Proposition \ref{prop:idealID}, $\dim(M_1 \cap M_2) < \min(\dim(M_1), \dim(M_2))$, as needed. 
\end{proof}

\bibliography{references.bib}{}
\bibliographystyle{plain}

\end{document}